\documentclass[%twocolumn,
amssymb, nobibnotes, notitlepage,nofootinbib aps, prd]{revtex4-1}

\usepackage{graphicx}
\usepackage{dcolumn}% Align table columns on decimal point
\usepackage{bm}% bold math
\usepackage{hyperref}% add hypertext capabilities
\input xy
\xyoption{all}

\newtheorem{lemma}{Lemma}[section]
\newtheorem{proposition}{Proposition}[section]
\newtheorem{theorem}{Theorem}[section]

\newtheorem{proof}{Proof}[section]

\usepackage{cleveref}
\crefname{thm}{theorem}{theorems}
\crefname{prop}{proposition}{propositions}
\crefname{cor}{corollary}{Corollary}
\crefname{defn}{definition}{Definition}
\crefname{defnb}{definition}{Definition}

\newcommand{\bea}{\begin{eqnarray}}
\newcommand{\eea}{\end{eqnarray}}
\newcommand{\beq}{\begin{equation}}
\newcommand{\eeq}{\end{equation}}

\newcommand{\vp}{\varphi}
\newcommand{\vpb}{\bar{\varphi}}
\newcommand{\mS}{\mathcal{S}}
\newcommand{\cG}{\mathcal{G}}
\newcommand{\mbC}{\mathbb{C}}

\newcommand{\mbZ}{\mathbb{Z}}
\begin{document}
\title{Beta Functions of $U(1)^d$ Gauge Invariant  Just Renormalizable Tensor Models }%
\author{Dine Ousmane Samary}%
\email[ ]{ousmanesamarydine@yahoo.fr,\,\,\,dine.ousmanesamary@cipma.uac.bj}
\affiliation{International Chair in Mathematical Physics and Applications\\
  ICMPA-UNESCO Chair, 072\,BP\,50, Cotonou,  B\'enin}
\date{April 02, 2013}%
\begin{abstract}
This manuscript reports the first order $\beta$-functions of recently proved just renormalizable random tensor models
endowed  with  a $U(1)^d$ gauge invariance [arXiv:1211. 2618]. The models that we consider are polynomial Abelian 
$\vp^4_6$ and $\vp^6_5$ models. We show in this work that  both  models are  asymptotically free in the UV. 
\end{abstract}
\maketitle
\tableofcontents
%%%%%%%%%%%%%%%%%%%%New Section%%%%%%%%%%%%%%%%%%%%
%%%%%%%%%%%%%%%%%%%%%%%%%%%%%%%%%%%%%%%%%%%%%%%
\section{Introduction}
Many interesting physical systems can be represented mathematically as random matrix problems. In particular, matrix models, celebrated in the 80's, provide a unique and well defined framework for addressing quantum gravity (QG) in two dimensions and its cortege of consequences on
integrable systems \cite{Di Francesco:1993nw}. The generalization of such models to higher dimensions is called random tensor models \cite{ambj3dqg}. Recently, these tensor models have acknowledged a strong revival thanks to the discovery by Gurau of the analogue of the t'Hooft $1/N$-expansion 
for the tensor situation \cite{Gurau:2010ba}-\cite{Gurau:2013cbh} and of tensor renormalizable actions  \cite{Ben-Geloun2011aa}-\cite{Samary:2012bw}. The tensor model framework begins to take a growing role in the problem of QG and raises as a true alternative to several known approaches \cite{Oriti:2006se,Rivasseau2011ab, Rivasseau2012ab}.

Tensorial group field theory (TGFT) \cite{Rivasseau2011ab, Rivasseau2012ab} is a recent proposal for the same problematic. It aims at providing a content to a phase transition called geometrogenesis scenario by relating a discrete quantum pre-geometric phase of our spacetime to the classical continuum limit consistent with Einstein gereral relativity. In short, within this approach, our spacetime and its geometry has to be reconstructed or must emerge from more fundamental and discrete degrees of freedom. 

Matrix models expand in graphs via ordinary perturbations of the Feynman path integral. These graphs can  be seen as dual to triangulations of two dimensional surfaces. Here, the discrete degrees of freedom refer to matrices, or more appropriately to their indices, or dually to triangles which glue to form a discrete version of a surface.  In tensor models, this idea generalizes. Feynman graphs in such tensor models are dual to triangulations of a $D$ dimensional object. 
The tensor field possesses discrete indices and it is dually related to a basic $D$
dimensional simplex which should be glued to others in order to form
a discretization of a $D$ dimensional manifold.

As for any quantum field theory, the question of renormalizability of TGFT 
has been addressed and solved under specific prescriptions \cite{Ben-Geloun2011aa}-\cite{Samary:2012bw}. Those conditions identify as 
the introduction of a Laplacian dynamics  for the action kinetic term \cite{Geloun:2011cy} and the use of non local interaction of the tensor invariant form \cite{Gurau:2011tj, Bonzom2012ac}.
Furthermore, as another important feature, the UV asymptotic freedom of some TGFTs has been proved in 3D \cite{Ben-Geloun2012aa} and 4D \cite{BenGeloun:2012yk} (see also \cite{Geloun:2012qn} for a shorter summary). This is strongly encouraging for the geometrogenesis
scenario. Indeed, the asymptotic freedom means that, after some scales towards the IR direction,
the renormalized coupling constant of the theory starts to blow up and, certainly, this entails a phase transition towards new degrees of freedom. This is analogue of the asymptotic freedom of non abelian Yang Mills theory leading to the better understanding of the quark confinement. However, the new degrees of freedom
in TGFTs have been not yet investigated.

New TGFT models, of the form of $\vp_6^4$ and $\vp_5^6$ theories, 
equipped with tensor fields obeying a gauge invariance condition were recently shown  just renormalizable at all orders of  perturbation \cite{Samary:2012bw}. The gauge invariant condition on tensor fields will help 
for the emergence of a well defined metric on the space after phase transition
\cite{Oriti:2006se,Carrozza2012aa}. 
The renormalization of the model followed from a multi-scale analysis and a generalized locality principle leading to a power-counting theorem 
\cite{Riv1}. 

In the present work, we calculate the first order $\beta$-function of both 
models and prove that these models are asymptotically free in
the UV regime. This paper also emphasizes that this asymptotic freedom could be a generic feature of all TGFTs for model with and without gauge invariance \cite{Rivasseau2012ab}. Such a feature will strengthen the status of TGFTs as pertinent candidates for gravity emergent scenario.  

The paper is organized as follows. We recall in section 2
the main results concerning the renormalizability of $\vp_6^4$ and $\vp_5^6$-tensor models as proved in \cite{Samary:2012bw}. 
 Section 3 is devoted to the study of the one-loop $\beta$-function 
of the $\vp_6^4$-model and section 4 addresses the computation of the same quantity, this time at higher order loops, for the $\vp_5^6$-model.
Finally, an appendix gathers technical points useful for the proof of our statements. 
%%%%%%%%%%%%%%%%%%%New section%%%%%%%%%%%%%%%%%%%
%%%%%%%%%%%%%%%%%%%%%%%%%%%%%%%%%%%%%%%%%%%%

\section{Abelian TGFT  with gauge invariance}

This section addresses a summary of the results obtained in   \cite{Samary:2012bw}.  We mainly present the model and its renormalization.

TGFTs over a group $G$ are defined by a complex field $\vp$ over  $d$ copies of group $G$, i.e.
\bea
\begin{array}{cccl}
\vp:  &G^d & \longrightarrow & \mbC \\
      & (g_1,\cdots, g_d) &\longmapsto & \vp(g_1,\cdots, g_d)\,.
\end{array}
\eea
The gauge invariance condition \cite{Oriti:2006se} is achieved by 
imposing that the fields obey the relation
\bea\label{Gauge-condition}
\vp(hg_{1},\dots,hg_{d})=\vp(g_{1},\dots,g_{d}),\quad \forall h\in G\,.
\eea
For Abelian TGFTs,  one fixes the group  $G=U(1)$.
In the momentum representation, the  field writes
\bea
&&\vp(g_1, \cdots, g_d)=\sum_{p}\vp_{[p]}e^{ip_1\theta_1}e^{ip_2\theta_2}\cdots e^{ip_d\theta_d},\quad  \theta_k\in [0,2\pi)\nonumber,
\eea
where we denote  $\vp_{[p]}=\vp_{12\cdots d}:=\vp(p_1,p_2,\cdots, p_d)$, with $p_k\in \mbZ$ and $g_k=e^{i\theta_k}\in U(1)$.

The generalized locality principle of the TGFTs considered in \cite{Samary:2012bw} requires to define the interactions as the sum of tensor invariants \cite{Gurau:2010ba}.  From now, we will focus on $d=6,5,$ and define two models described by
\bea
  \mathfrak{S}_{4}[\vpb,\vp]&=&\sum_{p_1,\cdots, p_6}\vpb_{654321}\,\delta(\sum_{i}^6p_{i})(p^{2}+m^{2})\,\vp_{123456}+
\textstyle{\frac12} \lambda^{(4)}_{4,1}\,V^6_{4,1},\label{eq:Action4}\\
  \mathfrak{S}_{6}[\vpb,\vp]&=&\sum_{p_1,\cdots, p_5}\vpb_{54321}\,\delta(\sum_{i}^5p_{i})(p^{2}+m^{2})\,\vp_{12345}\nonumber\\
  &+&\textstyle{\frac12} \lambda^{(6)}_{4,1}\,V^5_{4,1}+
\textstyle{\frac12} \lambda_{4,2}V_{4,2}+
\textstyle{\frac13} \lambda_{6,1}V_{6,1}+
\lambda_{6,2}V_{6,2},
\label{eq:Action6}
\eea
where $\delta(\sum_{i}^dp_{i})$ should be understood as a Kronecker symbol
$\delta_{\sum_{i}^dp_{i},0}$ and $p^2 = \sum_{i}^dp_i^2$, $d=6,5,$ respectively,   and where the interactions are of the  form given by

\begin{figure}[htbp]
\begin{center}
{{\label{ClosedEx}}\includegraphics[scale=0.7]{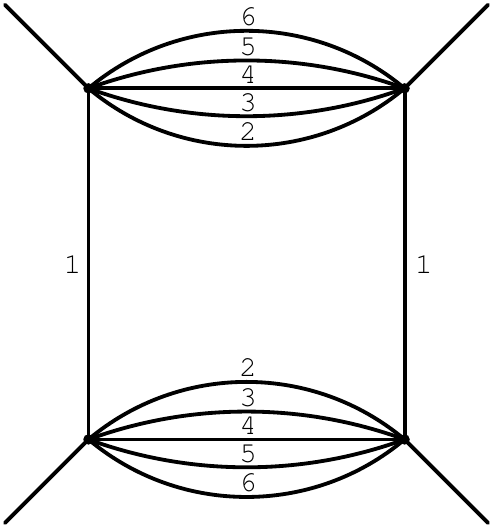}}\hspace{2cm}
\end{center}
 \caption{Vertex representation of $\vp_6^4$-model}
  \label{fig:Vertex46} 
\end{figure}

\begin{figure}[htbp]
% \begin{center}
$A=V_{6,1}${ {\label{ClosedExttt}}\includegraphics[scale=0.5]{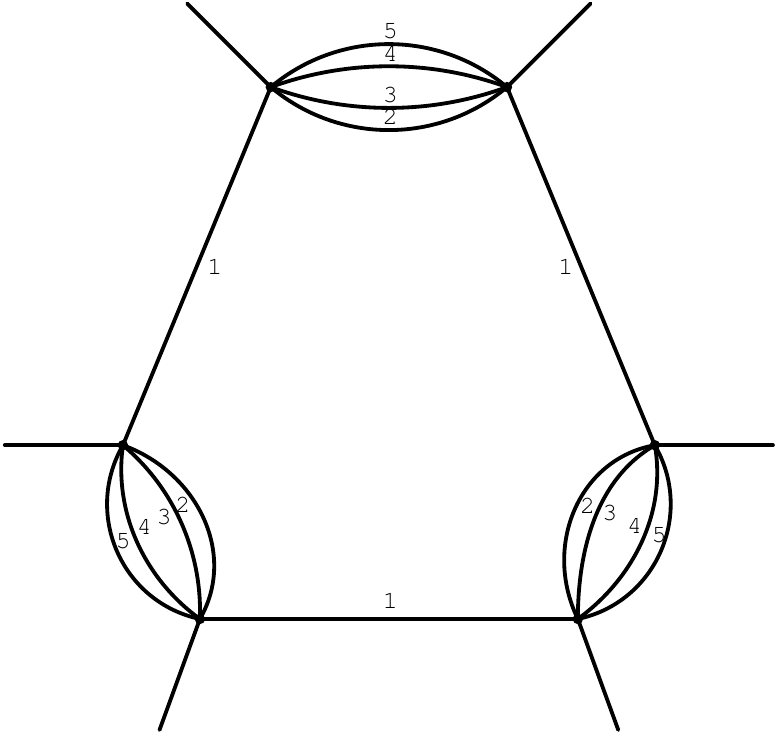}}\hspace{3cm}
  $B=V_{6,2}$ {{\label{OpenExttt}}\includegraphics[scale=0.6]{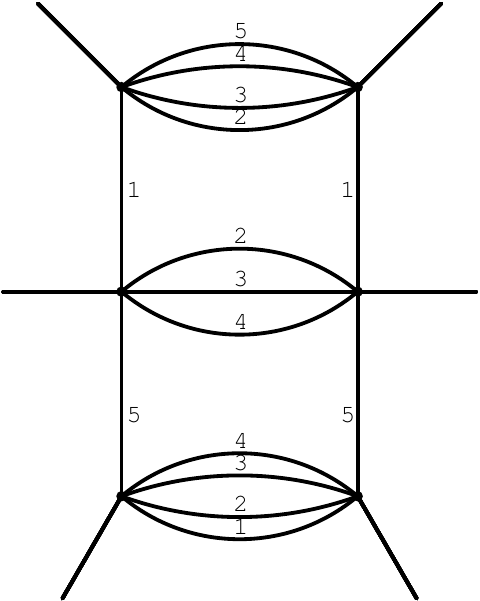}}\\
$C=V_{4,1}^5$ {{\label{OpenExttt}}\includegraphics[scale=0.6]{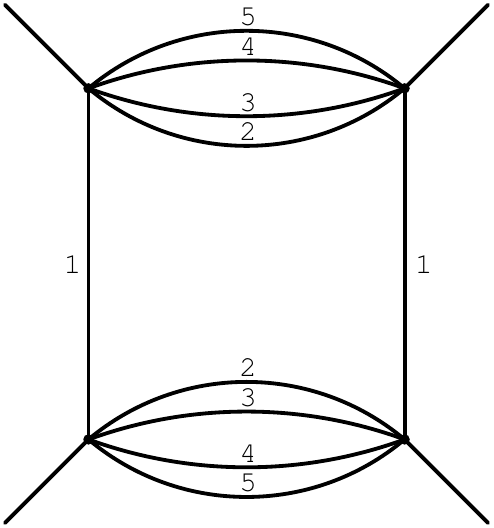}}\hspace{3cm}
 $D=V_{4,2}$ {{\label{ClosedExjjj}}\includegraphics[scale=0.6]{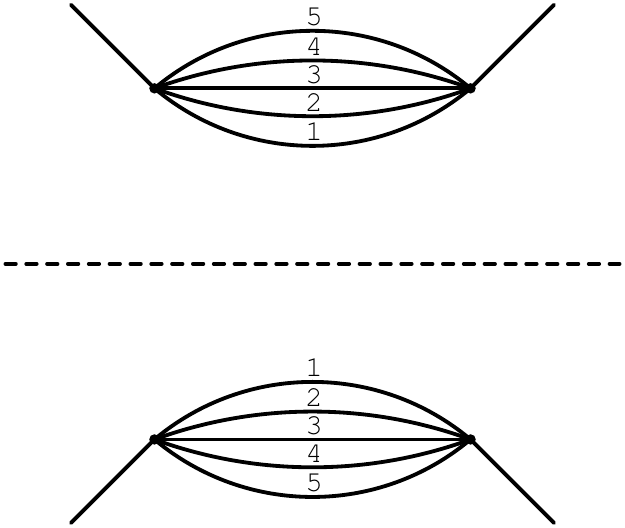}}
%\end{center}
  \caption{Vertex representation of $\vp_5^6$-model}
  \label{fig:Vertex65}
\end{figure}

\begin{widetext}
\bea
  \label{eq:vertex}
  V^{6}_{4,1}&=&\sum_{\mathbb{Z}^{12}}\vpb_{654321}\,\vp_{12'3'4'5'6'}\,\vpb_{6'5'4'3'2'1'}\,\vp_{1'23456}+\mbox{
  permutations} ,\\
\label{eq:vertex2}
 V^{5}_{4,1}&=&  \sum_{\mathbb{Z}^{12}}\vpb_{54321}\,\vp_{12'3'4'5'}\,\vpb_{5'4'3'2'1'}\,\vp_{1'2345}+\mbox{
  permutations},\\
 V_{4,2}&=&\big(\sum_{\mathbb{Z}^{5}}\vpb_{54321}\vp_{12345}\big)^{2},\\
  V_{6,1}&=&\sum_{\mathbb{Z}^{15}}\vpb_{54321}\vp_{1'2345}\vpb_{5'4'3'2'1'}\vp_{1"2'3'4'5'}\vpb_{5"4"3"2"1"}\vp_{12"3"4"5"}+\mbox{
    permuts}\label{eq:V61},\\
  V_{6,2}&=&\sum_{\mathbb{Z}^{15}}\vpb_{54321}\vp_{1'2345}\vpb_{5'4'3'2'1'}\vp_{1"2"3"4"5'}\vpb_{5"4"3"2"1"}\vp_{12'3'4'5"}+\mbox{
    permutations.}\label{eq:V62}
\eea
\end{widetext}

The ``permutations'' are performed on the color indices. 
The vertices are graphically represented in \cref{fig:Vertex46} and \cref{fig:Vertex65}. 
As one notices, there is two kinds of lines in the  vertices. 
The first type are parametrized by $1,2,\dots,d$ and one external half-line
without any number. Call by 0 the color of this half-line.

The propagator of each model reads:  
\bea
C([p])=\frac{1}{\sum_{i=1}^dp_i^2+m^2}\delta(\sum_{i=1}^dp_i), \qquad 
d=6,5,
\eea
and it is represented graphically as a line with $d$ strands, see \cref{fig:propa}.

A Feynman graph is a graph composed with lines of color 
0 (propagators) and vertices. Hence, whenever we refer to a line
in the following it will be always a 0-color line and $\cG$ is an uncolored tensor graph in the sense of \cite{Gurau:2010ba} and \cite{Bonzom2012ac} which have $d$-strand lines  of color $0$.

Let $\mathcal{L}$ and $\mathcal{F}$ be the sets of internal lines and faces of the   graph  $\cG$.
The multi-scale analysis shows that the divergence degree of the amplitude of  a graph associated with both models can be written 
\beq
\omega_d(\cG)=2L-F+R
\eeq
 where $L= |\mathcal{L}|$, $F= |\mathcal{F}|$ and  $R$ is the rank of matrix $\big(\epsilon_{lf},\, l\in\mathcal{L},\, f\in\mathcal{F}\big),$ defined by
 \bea
    \epsilon_{lf}(\cG)=\left\{\begin{array}{ccc}
     1&\mbox{ if $l\in f$ and their orientation match,}\\
      -1&\mbox{ if $l\in f$ and their orientation do not match,}\\
     0&\mbox{otherwise.}
   \end{array}\right.
\label{eq:Epslf}
  \eea
\begin{figure}[htbp]
{{\label{ClosedEx}}\includegraphics[scale=0.6]{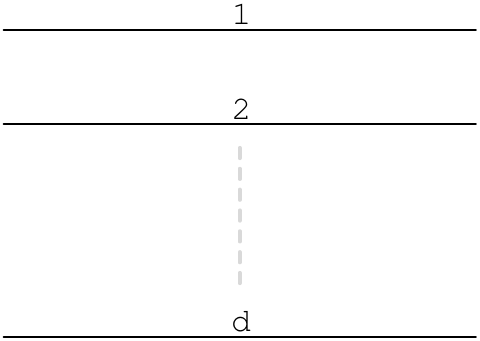}}\hspace{2cm}
  \caption{Propagator of $d$-dimensional tensor model}
  \label{fig:propa}
\end{figure}

The following statement  holds \cite{Samary:2012bw}:

\begin{theorem}
  \label{thm-PertRen}
 The models $\vp^4_6$ defined  by $ \mathfrak{S}_{4}$
and $\vp^6_5$ defined  by $ \mathfrak{S}_{6}$ are perturbatively renormalizable at all orders. 
\end{theorem}
The proof of this statement rests on a 
power counting theorem which can be summarized by the following table giving the list of primitively divergent graphs (for precisions and notations, see \cite{Samary:2012bw}):
 \begin{eqnarray}
  \begin{array}{|c|c|c|c|c||c|}
\hline
    &N&\omega(\cG)&\omega(\partial\cG)&C_{\partial\cG}-1&\omega_{d}(\cG)\\
    \hline
    %\hline
    \vp^{4}_{6}&4&0&0&0&0\\
    &2&0&0&0&2\\
%&2&0&0&0&1\\
    \hline
    \vp^{6}_{5}&6&0&0&0&0\\
    &4&0&0&0&1\\
    &4&0&0&1&0\\
    &2&0&0&0&2\\
 &2&0&0&0&1\\
    \hline
    \hline
  \end{array}
\\
\mbox{ Table: Divergent graphs of models \ref{eq:Action4} and \ref{eq:Action6}}\nonumber
\end{eqnarray}

 with $N$  the number of external fields, $\omega(\cG)$ the degree of the colored extension of the graph $\cG$, $\omega(\partial\cG)$ the degree
of the colored extension of the boundary $\partial\cG$  of the graph $\cG$, 
$C_{\partial\cG}$ the number of connected component of the boundary graph 
$\partial \cG$.

Using this table, we are now in position to compute renormalized coupling equations.

%%%%%%%%%%%%%%New section%%%%%%%%%%%%%%%%%%%%%%%
%%%%%%%%%%%%%%%%%%%%%%%%%%%%%%%%%%%%%%%%%%%

\section{One-loop $\beta$-function of $\vp_6^4$-model}
This section is devoted to the  one-loop evaluation 
of the $\beta$-function of $\vp_6^4$. To proceed, we enlarge the space of coupling constants so that (\ref{eq:Action4})  becomes
\bea
 \mathfrak{S}_{4}[\vpb,\vp]&&\sum_{p_1,\cdots p_6}\vpb_{654321}\,\delta(\sum_{i}^6p_{i})(p^{2}+m^{2})\,\vp_{123456}+
\textstyle{\frac12}\sum_{\rho=1}^6 \lambda_{4,1;\rho}\,V^6_{4,1;\rho}.
\eea
Only at the end we will perform a merging of all coupling at the same value $\lambda_{4,1;\rho}=\lambda_{4,1}$.  Thus by introducing a distinction between
the colors, $\rho=1,2,\dots,6$, the combinatorics becomes less involved. 

We have the following theorem:
\begin{theorem}\label{th1}
At one-loop, the renormalized coupling constant associated with $\lambda_{4}$ is given by
\bea
\lambda_{4}^{\rm ren}=\lambda_{4}+\frac{19\pi^2}{5\sqrt{5}}\lambda_{4}^2\,\mathcal{I} +O(\lambda_{4}^2),\quad\mbox{with }\quad \mathcal{I}=\int_0^\infty\frac{e^{-\alpha m^2}}{\alpha}\,d\alpha
\eea
such that the $\beta$-function of the model with single wave-function renormalization and single coupling constant is given by
$\beta=-\frac{19\pi^2}{5\sqrt{5}}$.
\end{theorem}
We now prove   Theorem \ref{th1}.
Let $Z$  be the wave function  renormalization which writes:
\bea\label{eq:zdep}
Z=1-\frac{\partial^2}{\partial b_\rho^2}\Sigma\Big|_{[b]=0}, \qquad \rho=1,2,\cdots,6, 
\eea
where  $\Sigma$ is called the  self-energy or the sum of all amputated
one-particle irreducible (1PI) two-point functions which must  be evaluated at one-loop. 
The derivative on $\Sigma$ is with respect to an external 
argument. 
The $\beta$-function of the model $\vp^4_6$ is encoded by 
the following quotient
\bea
\lambda^{\rm ren}_{4} = - \frac{\Gamma_{4}(0)}{Z^2}
\eea
where $\Gamma_{4}$ is the sum of all amputated 1PI four-point functions
computated at one-loop and at low external momenta
that we symbolize by a unique argument $(0)$. 
\medskip 

\noindent{\bf Self-energy and wave function renormalization.}
 Having a look on (\ref{eq:zdep}) only is relevant the dependance
in some color $\rho$ of $\Sigma$. 
We will evaluate only this part in the self-energy at one-loop. 
For two sets of external arguments $[b]$ and 
$[b']$, one has
\bea\label{topointsnew}
\Sigma([b],[b'])=<\bar\vp_{[b]}\vp_{[b']}>^t_{1PI}=\sum_{\cG}K_{\cG}\,A_{\cG}([b],[b'])
\eea
where $K_\cG$ is a combinatorial factor and $A_{
\cG}$ is the amplitude  of the graph $\cG$. 
Let 
\bea
\mS(b)=\sum_{p_1,\cdots,p_4}\big[\big(\sum_{k=1}^4 p_k^2\big)+\big(\sum_{k=1}^4 p_k\big)^2+2b\sum_{k=1}^4 p_k+2b^2+m^2\big]^{-1}
\eea
\bea
\mS'(b)=\sum_{p_1,\cdots,p_4}\big[\big(\sum_{k=1}^4 p_k^2\big)+\big(\sum_{k=1}^4 p_k\big)^2+2b\sum_{k=1}^4 p_k+2b^2+m^2\big]^{-2}
\eea
\bea
\mathcal{K}(b)=\sum_{p_1,\cdots,p_4}\frac{\big(\sum_{k=1}^4 p_k+2b\big)^2}{\big[\sum_{k=1}^4 p_k^2+\big(\sum_{k=1}^4 p_k\big)^2+2b\sum_{k=1}^4 p_k+2b^2+m^2\big]^3}.
\eea 
%For the rest the dimension $d=6$ and ensure the renormalization of the model \cite{Samary:2012bw}.

 At one-loop, there exist six tadpole graphs $T_\rho,\,\rho=1,\cdots, 6,$  that contribute to
the relation  (\ref{topointsnew}). For instance $T_1$ is represented in  \cref{fig:Tadpole}.
\begin{figure}[htbp]
{{\label{ClosedEx}}\includegraphics[scale=0.5]{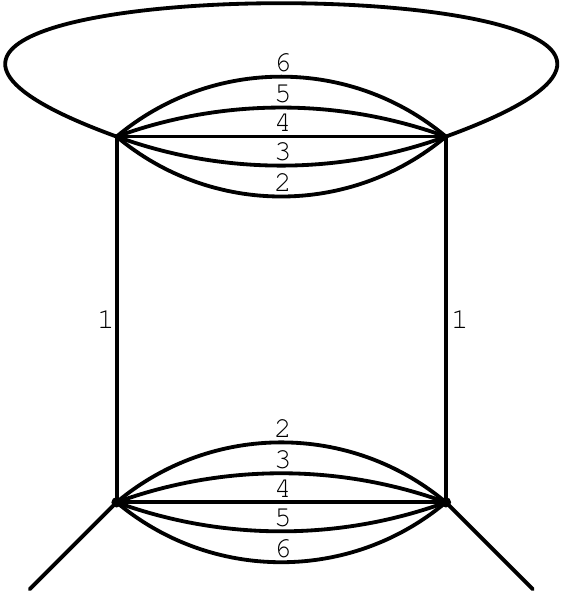}}
  \caption{Tadpole  graphs $T_1$}
  \label{fig:Tadpole}
\end{figure}
The amplitude associated to the tadpole $T_\rho$ is given by
\bea
A_{T_\rho}
&=&-\frac{\lambda_{4,\rho}}{2}\mS(b_\rho).
\eea
The combinatorial weight of these  graphs $T_\rho$ is $K_{T_\rho}=2$. Then (\ref{topointsnew}) is re-expressed as
\bea
 \Sigma([b])=-\sum_{\rho=1}^6\lambda_{4,\rho}\mS(b_\rho).
\eea
We have the following relation (see Appendix I for details):
\bea\label{zoulou}
\mS'(0)=\frac{\pi^2}{\sqrt{5}}\mathcal{I},\quad \mathcal{K}(0)=\frac{\pi^2}{5\sqrt{5}}\mathcal{I},\quad \mathcal{I}=\int_0^\infty\,d\alpha \frac{e^{-\alpha m^2}}{\alpha},
\eea
then
\bea
\frac{\partial^2 \Sigma[b]}{\partial b_\rho^2}\Big|_{[b]=0}=4\lambda_{4,\rho}\Big(\mS'(b_\rho)-2
\mathcal{K}(b_\rho)\Big)\Big|_{[b]=0}=\frac{12\pi^2}{5\sqrt{5}}\lambda_{4,\rho}\mathcal{I}.
\eea
Using the fact that the tadpole amplitudes   are symmetric  with respect to the external variables, we  reduce all coupling constants to the same value i.e. $\lambda_{4,\rho}=\lambda_4,$ and get  
 the wave function renormalization  as
\bea
Z=1-\frac{12\pi^2}{5\sqrt{5}}\lambda_{4}\,\,\mathcal{I} +O(\lambda_4^2).
\eea

\noindent{\bf Four-point functions.}
The 1PI four-point function amplitudes $\Gamma_{4,\rho},\, \rho=1,2,\cdots, 6,$ are given by
\bea\label{Four-point}
\Gamma_{4,\rho}([b], [b'])=<\bar\vp_{[b]_1}\vp_{[b]_2}\bar\vp_{[b']_1}
\vp_{[b']_2}>^t_{1PI}=\sum_{\cG}K_{\cG} A_{\cG}([b], [b']),
\eea
where $[b]_j,[b']_j,\, j=1,2,$ are the external strand indices.
Using the cyclic permutation over the six indices $\rho$, the  four-point functions  are explicitly given by
\bea
\Gamma_{4,1}(b_1,\cdots, b_6,b_1'\cdots,b_6')=<\bar\vp_{123456}\vp_{65'4'3'2'1'}\bar\vp_{1'2'3'4'5'6'}\vp_{6'54321}>^t_{1PI}\\
\Gamma_{4,2}(b_1,\cdots, b_6,b_1'\cdots,b_6')=<\bar\vp_{123456}\vp_{6'54'3'2'1'}\bar\vp_{1'2'3'4'5'6'}\vp_{65'4321}>^t_{1PI}\\
\Gamma_{4,3}(b_1,\cdots, b_6,b_1'\cdots,b_6')=<\bar\vp_{123456}\vp_{6'5'43'2'1'}\bar\vp_{1'2'3'4'5'6'}\vp_{654'321}>^t_{1PI}\\
\Gamma_{4,4}(b_1,\cdots, b_6,b_1'\cdots,b_6')=<\bar\vp_{123456}\vp_{6'5'4'32'1'}\bar\vp_{1'2'3'4'5'6'}\vp_{6543'21}>^t_{1PI}\\
\Gamma_{4,5}(b_1,\cdots, b_6,b_1'\cdots,b_6')=<\bar\vp_{123456}\vp_{6'5'4'3'21'}\bar\vp_{1'2'3'4'5'6'}\vp_{65432'1}>^t_{1PI}\\
\Gamma_{4,6}(b_1,\cdots, b_6,b_1'\cdots,b_6')=<\bar\vp_{123456}\vp_{6'5'4'3'2'1}\bar\vp_{1'2'3'4'5'6'}\vp_{654321'}>^t_{1PI}
\eea
\begin{figure}[htbp]
{{\label{ClosedEx}}\includegraphics[scale=0.3]{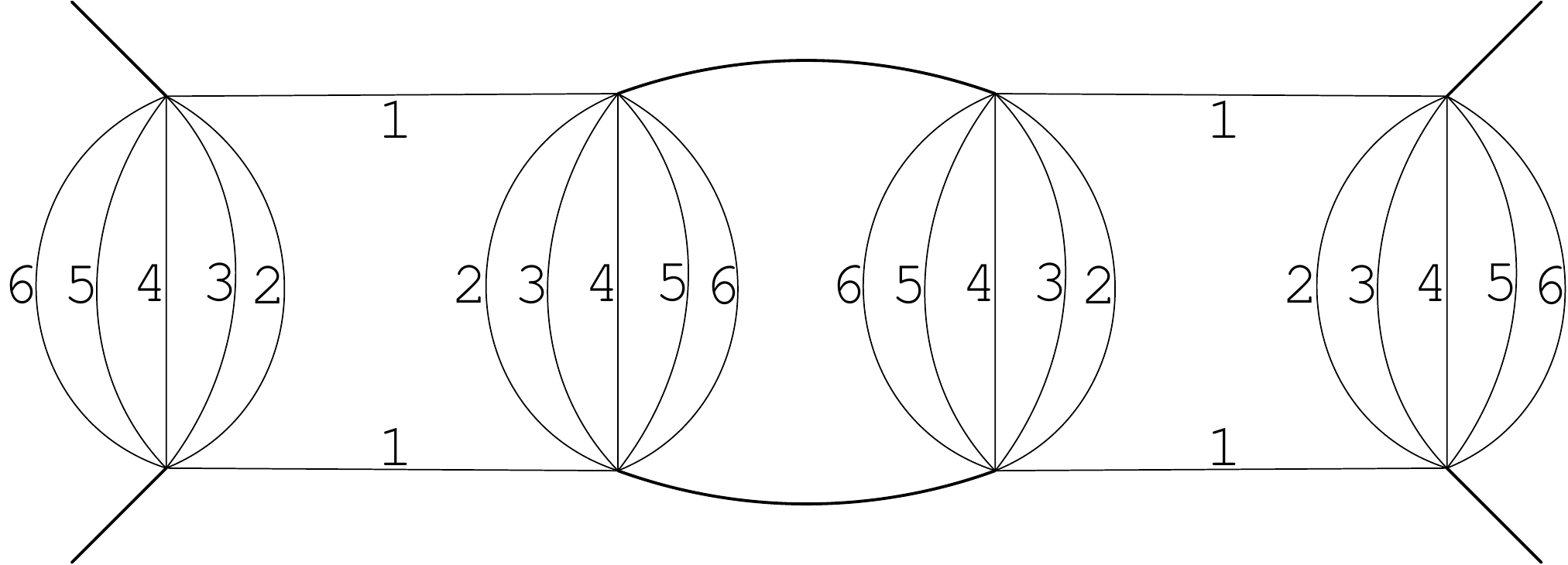}}\hspace{2cm}
  \caption{The melonic one-loop four-point  graph}
  \label{fig:fourpoint}
\end{figure}
At one-loop,  there is a unique graph contributing to $\Gamma_{4;\rho}$. It is of the form given by \cref{fig:fourpoint}.
The combinatorial factor of this graph is always  $K_{\cG}=2\cdot 2\cdot 2$. The amplitude associated of this graph
 is
\bea
A_{\cG,\epsilon}(b)=\frac{\lambda_{4,\epsilon}^2}{2^2.2}\mS'(b).
\eea
We obtain
\bea
\Gamma_{4}(0)&=&-\lambda_{4}+
\lambda_{4}^2\mS'(0)
+O(\lambda_4^2)=-\lambda_{4}
+\frac{\pi^2}{\sqrt{5}}\lambda_{4}^2\,\mathcal{I}+O(\lambda_4^2).
\eea
The renormalizable coupling constant is finally given by
\bea\label{crach}
\lambda^{{\rm ren}}_{4}&=&-\frac{\Gamma_{4}(0,0)}{Z^2}=\lambda_{4}+\frac{19\pi^2}{5\sqrt{5}}\lambda_{4}^2\,\mathcal{I} +O(\lambda_4^2).
\eea
This result shows that the $\vp_6^4$ model is asymptotically free in the UV regime. The $\beta$-function at one-loop of the model reads from (\ref{crach}): 
\bea
\beta=-\frac{19\pi^2}{5\sqrt{5}}.
\eea

%%%%%%%%%%%%%%%%%%%%New section%%%%%%%%%%%%%%%%%%%%
%%%%%%%%%%%%%%%%%%%%%%%%%%%%%%%%%%%%%%%%%%%%%%

\section{Two-loop $\beta$-functions of the $\vp_5^6$-model}
\begin{figure}[htbp]
$V_{6,1}${{\label{ClosedExttt}}
\includegraphics[scale=0.5]{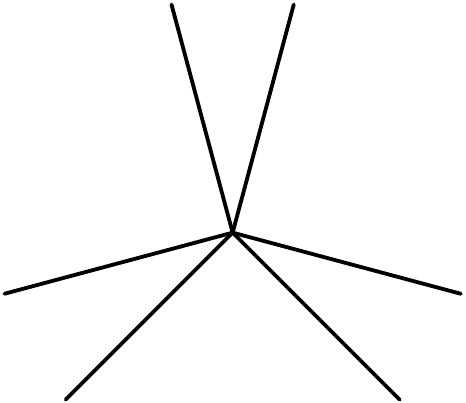}}\hspace{1.5cm}
$V_{6,2}${{\label{ClosedExttt}}
\includegraphics[scale=0.5]{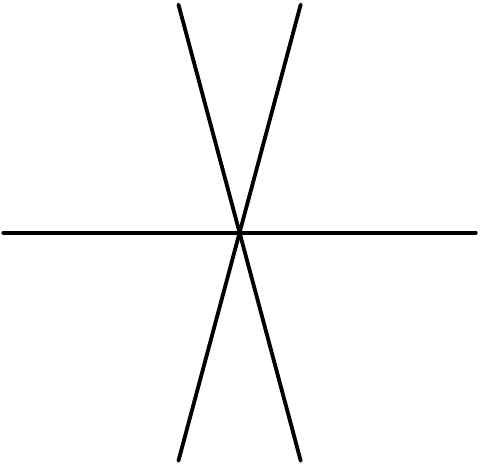}}\hspace{1.5cm}
  \caption{New graphical representation of vertices  $V_{6,1}$ and $V_{6,2}$}
  \label{fig:re-vertexe}
\end{figure}
In the $\vp_5^6$-model, there are two types of coupling constants and so we must evaluate two renormalized coupling equations.
In order to  compute the $\beta$-functions of the $\vp_5^6$ model it is important to note that the vertices of the type $V_{6,1}$ are parametrized by five indices $\rho=1,2,\cdots,5$,  and  the vertices contributing to $V_{6,2}$ are parametrized by  ten  indices  $\rho\rho'=1.2, 1.3,1.4, 1.5, 2.3,2.4,2.5, 3.4,3.5,4.5$.  The couple $\rho\rho'$ will be totally    symmetric i.e., $\rho\rho'=\rho'\rho$. For 
simplicity, the graphs of \cref{fig:re-vertexe}  represent henceforth the vertices of  $\vp_5^6$ model. For the same combinatorial reasons evoked above, we enlarge again the space of coupling and write (\ref{eq:Action6}) as
\bea
 \mathfrak{S}_{6}[\vpb,\vp]&=&\sum_{p_1,\cdots, p_5}\vpb_{54321}\,\delta(\sum_{i}^5p_{i})(p^{2}+m^{2})\,\vp_{12345}+\frac{1}{3}\sum_{\rho}\lambda_{6,1;\rho}V_{6,1;\rho}\cr
&+&\sum_{\rho\rho'}\lambda_{6,2;\rho}V_{6,2;\rho\rho'}+\frac{1}{2}\sum_{\rho}\lambda_{4,1;\rho}V_{4,1;\rho}+\frac{1}{2}\sum_{\rho}\lambda_{4,2}V_{4,2}.
\eea
We have the following theorem:
\begin{theorem}\label{prop3}
The  renormalized coupling 
 constants $\lambda_{6,1}^{\rm ren}$ and $\lambda_{6,2}^{\rm ren}$ satisfy the equations  
\bea
\lambda_{6,1}^{\rm ren}
&=&\lambda_{6,1}+\frac{9\pi^3}{4}\lambda_{6,1}^2\,\mathcal{I'}+12\big(\frac{49}{31\sqrt{31}}+\frac{5}{8}\big)\pi^3\lambda_{6,1}\lambda_{6,2}\mathcal{I'}
+ O(\lambda_{6,1}^p\lambda_{6,2}^{3-p}),
\eea
 and
\bea
\lambda_{6,2}^{\rm ren}
&=&\lambda_{6,2}+4\Big(\frac{178}{31\sqrt{31}}+\frac{11}{8}\Big)\pi^3\lambda_{6,2}^2\,\mathcal{I'}+\frac{11\pi^3}{4}\lambda_{6,1}\lambda_{6,2}\,\mathcal{I'}
+ O(\lambda_{6,1}^p\lambda_{6,2}^{3-p}),
\eea
$p=0,1,2,3$.
\end{theorem}
\noindent{\bf Self-energy and wave function renormalization.}
The following proposition holds:
\begin{proposition}
 The  wave function renormalization of the model is 
given by
\bea
Z&=&1-\frac{5\pi^3}{4}\lambda_{6,1}\,\mathcal{I'}-4\big(\frac{80}{31\sqrt{31}}+\frac{5}{8}\big)\pi^3\lambda_{6,2}\,\mathcal{I'}+O(\lambda_{6,1}^p\lambda_{6,2}^{2-p}),
\eea
$p=0,1,2$, and where $\mathcal{I'}$ writes
\bea
\mathcal{I'}=\int_0^\infty\,\int_0^\infty\,d^2\alpha\frac{e^{-2\alpha m^2}}{\alpha^2}.
\eea
\end{proposition}
{\bf proof:} Let us consider the following series

\bea
\mS^1(b)&=&\sum_{\stackrel{p_1,p_2,p_3}{q_1,q_2,q_3}}\Big\{\big[\sum_{k=1}^3 p_k^2+(\sum_{k=1}^3 p_k)^2+2b\sum_{k=1}^3 p_k+2b^2+m^2\big]^{-1}\cr
&&\times \big[\sum_{k=1}^3 q_k^2+(\sum_{k=1}^3 q_k)^2+2b\sum_{k=1}^3 q_k+2b^2+m^2\big]^{-1}\Big\},
\eea
\bea
\mS^{12}(b)&=&\sum_{\stackrel{p_1,p_2,p_3,p_4}{q_1,q_2}}\Big\{\big(\sum_{k=1}^4 p_k^2+(\sum_{k=1}^4 p_k)^2+m^2\big)^{-1}\big(\sum_{k=1}^2 q_k^2+(\sum_{k=1}^2 q_k)^2+2b\sum_{k=1}^2 q_k\cr
&+&2(\sum_{k=1}^4 p_k)^2-2b\sum_{k=1}^4 p_k-2\sum_{k=1}^4 p_k\sum_{k=1}^2 q_k+2b^2+m^2\big)^{-1}\Big\},
\eea
and
\bea
\mS^{13}(b,b')&=&\sum_{\stackrel{p_1,p_2,p_3}{q_1,q_2,q_3}}\Big\{\big[\sum_{k=1}^3 p_k^2+(\sum_{k=1}^3 p_k)^2+2b\sum_{k=1}^3 p_k+2b^2+m^2\big]^{-1}\cr
&&\times \big[\sum_{k=1}^3 q_k^2+(\sum_{k=1}^3 q_k)^2+2b'\sum_{k=1}^3 q_k+2b'^2+m^2\big]^{-1}\Big\}.
\eea

The graphs contributing to the self-energy are of the form  listed in  \cref{fig:Tadpolesee}.
\begin{figure}[htbp]
$T_{1;\rho}${{\label{ClosedExttt}}
\includegraphics[scale=0.5]{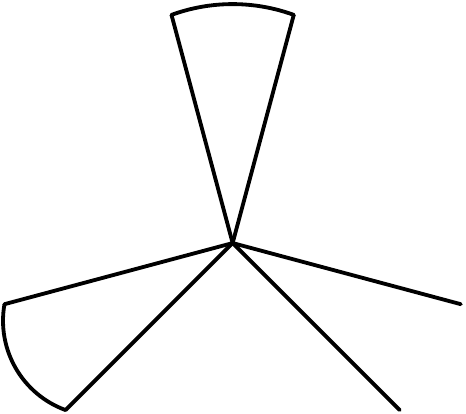}}\hspace{0.5cm}
$T_{2;\rho\rho'}$ {{\label{OpenExttt}}
\includegraphics[scale=0.5]{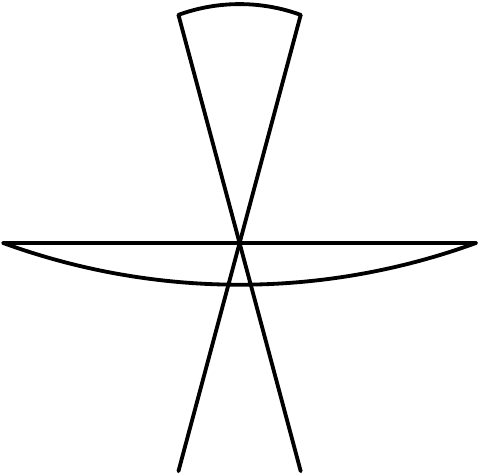}} \hspace{0.5cm}
$T_{3;\rho\rho'}$ {{\label{OpenExttt}}
\includegraphics[scale=0.5]{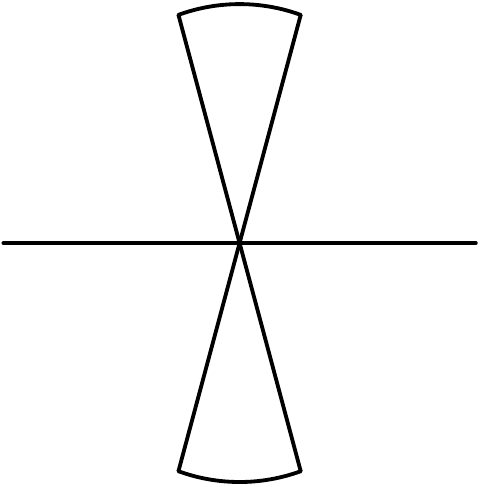}}\hspace{1.9cm}
  \caption{Divergent tadpoles graphs of $\vp_5^6$-model}
  \label{fig:Tadpolesee}
\end{figure}
 The amplitude corresponding to the tadpoles graphs $T_{1;\rho}$ is given by the following relation
\bea
A_{T_{1;\rho}}(b_\rho)&=&-\frac{\lambda_{6,1;\rho}}{3}\mS^1(b_\rho).
\eea
In the above expression $b_\rho$ is  an   external strand index.
Using  the combinatorial number associated to the tadpole graph $T_{1;\rho}$ given by $K_{T_{1;\rho}}=3$, the sum of  1PI two-point functions  are given by
\bea\label{topoints}
\Omega_{6,1}(b_\rho)=3A_{T_{1;\rho}}(b_\rho).
\eea
Similarly, the amplitude corresponding to the tadpole graphs $T_{2;1\rho}$ and $T_{3;1\rho}$ are respectively given by  relations
\bea
A_{T_{2;1\rho}}(b_{1})&=&-\lambda_{6,2;1\rho}\,\mS^{12}(b_1),
\eea
and
\bea
A_{T_{3;1\rho'}}(b_1,b_{\rho})&=&-\lambda_{6,2;1\rho}\,\mS^{13}(b_1,b_{\rho}).
\eea
The combinatorial factors are  $K_{T_{2;1\rho}}=1$ and $K_{T_{3;1\rho}}=1$. Therefore the  sum of these contribution yields
\bea\label{totopoints}
\Omega_{6,2}(b_1,b_{\rho})=A_{T_{2;1\rho}}(b_1)+A_{T_{3,1\rho}}(b_1,b_{\rho}).
\eea
 Combining the relations (\ref{topoints}) and (\ref{totopoints}),  we get a part of the self-energy involving the variable $b_1$ 
\bea
\Sigma_6(b_1,b_{\rho})=3A_{T_{1;1}}(b_1)+\sum_{\rho}\Big[A_{T_{2;1\rho}}(b_1)+A_{T_{3;1\rho}}(b_1,b_{\rho})\Big] + O(\lambda_{6,1}^p\lambda_{6,2}^{2-p}). 
\eea
 The wave function renormalization of the model is given by
\bea
&Z=1-\frac{\partial^2}{\partial b_1^2}\Sigma_6(b_1,b_{\rho})\big|_{b_1=b_{\rho}=0}.
\eea
Using  appendix \ref{serie2}, we have the following relations:
\bea
\frac{\partial^2}{\partial b_1^2}\Omega_{6,1}(b_1)\big|_{b_1=0}=\frac{5\pi^3}{4}\lambda_{6,1;1}\,\mathcal{I'},\quad \mathcal{I'}=\int_0^\infty\,\int_0^\infty\, d^2\alpha \,\frac{e^{-2\alpha m^2}}{\alpha^2}
\eea
\bea
\frac{\partial^2}{\partial b_1^2}\Omega_{6,2}(b_1,b_{\rho})\big|_{b_1=b_{\rho}=0}&=&
\big(\frac{80}{31\sqrt{31}}+\frac{5}{8}\big)\pi^3\lambda_{6,2,1\rho}\,\mathcal{I'}.
\eea
We restrict from now the coupling constants in each sector such that
$\lambda_{6,1;\rho} = \lambda_{6,1}$ and 
$\lambda_{6,2;\rho\rho'} = \lambda_{6,2}$ so that,
the wave function renormalization is 
\bea
Z&=&1-\frac{5\pi^3}{4}\lambda_{6,1}\,\mathcal{I'}-4\big(\frac{80}{31\sqrt{31}}+\frac{5}{8}\big)\pi^3\lambda_{6,2}\,\mathcal{I'}+ O(\lambda_{6,1}^p\lambda_{6,2}^{2-p}),
\eea
$p=0,1,2$. $\square$

\noindent{\bf Six-point functions.}
 The initial calculation of the six-point functions shows that they
prolifer quickly \cite{BenGeloun:2012yk}. However, in the present
gauge invariant model which is more constrained, several of these
should be not renormalized because either are convergent (pay attention to the
fact that gauge invariant models are less divergent than the ordinary one) or turn out to violate the face-connectedness condition (see discussion below and \cref{fig:6-new2}) \cite{Carrozza2012aa, Samary:2012bw}. 

At the end, we will focus on the six-point functions which are face-connected graphs of type $V_{6,1}-V_{6,1}$ and  $V_{6,1}-V_{6,2}$, see \cref{fig:6-new}. This will be used for the calculation of the sum of 1PI six-point functions $\Gamma_{6,1;\rho}$ and $\Gamma_{6,2;\rho\rho'}$.  The renormalized coupling constant equations for $\lambda_{6,1;1}^{\rm ren}$ and $\lambda_{6,2;1\rho'}^{\rm ren}$ are defined by
\bea
\lambda_{6,1;\rho}^{\rm ren}=-\frac{\Gamma_{6,1;\rho}(0,0)}{Z^3},\quad \lambda_{6,2;\rho\rho'}^{\rm ren}=-\frac{\Gamma_{6,2;\rho\rho'}(0,0)}{Z^3}.
\eea

{\bf Proof of Theorem \ref{prop3}}

The first part of this proof is about the evaluation of  amplitudes of various graphs of \cref{fig:6-new}.
\begin{figure}[htbp]
$G_{1;\rho}${{\label{ClosedExttt}}\includegraphics[scale=0.4]{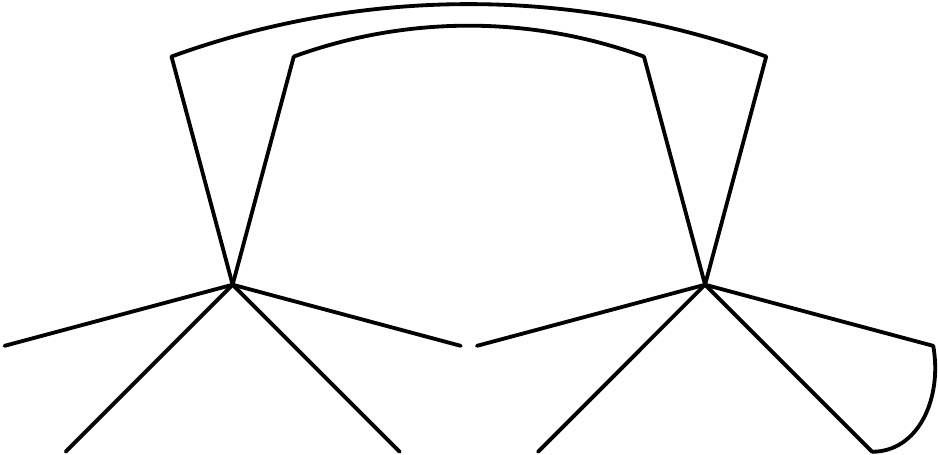}}\hspace{0.5cm},$\qquad$
$G_{2;\rho\rho'}${{\label{ClosedExttt}}\includegraphics[scale=0.4]{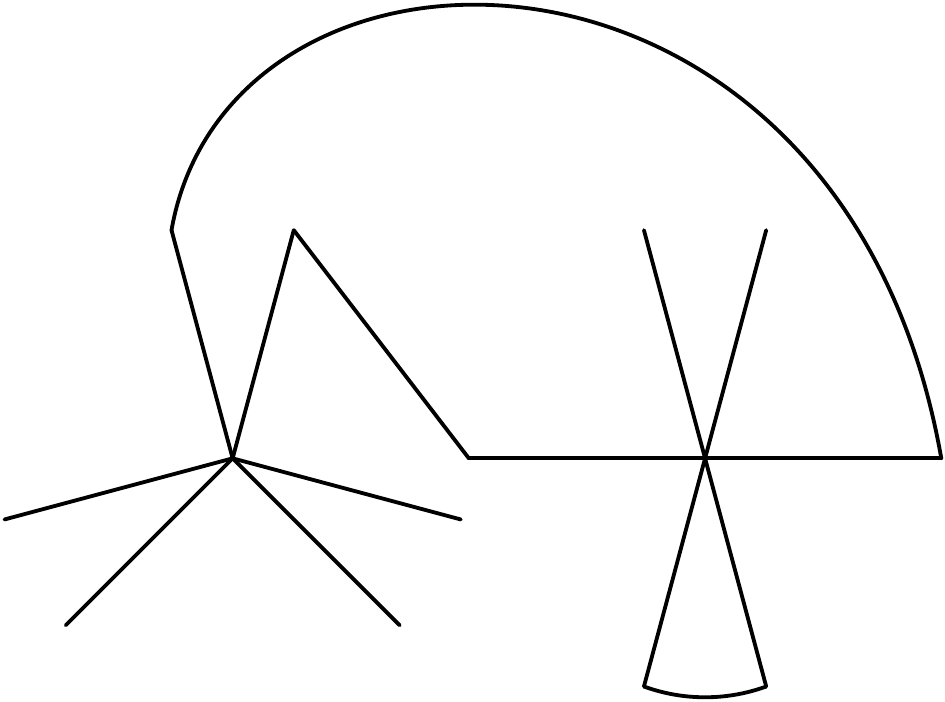}}\hspace{0.5cm}\\
$G'_{2;\rho\rho'}${{\label{ClosedExttt}}\includegraphics[scale=0.4]{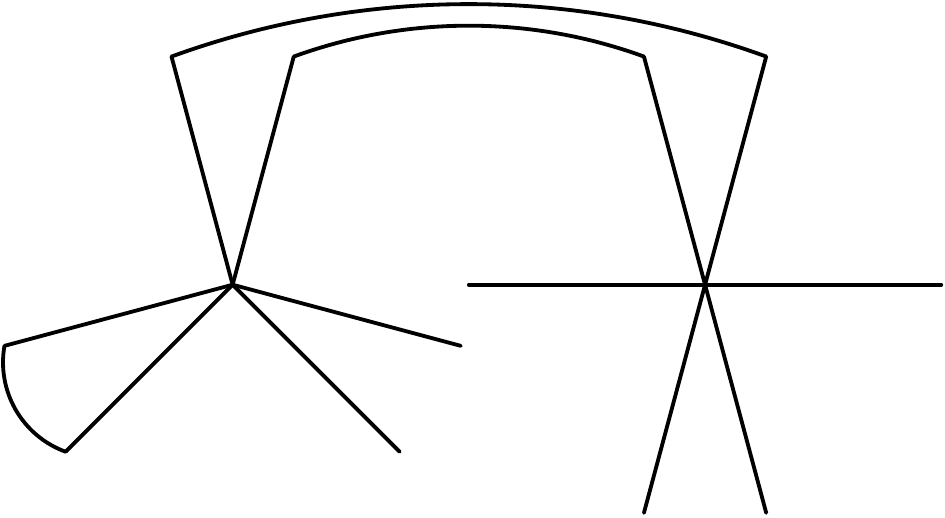}}\hspace{0.5cm},$\qquad$
$G_{3;\rho\rho'}$ {{\label{OpenExttt}}\includegraphics[scale=0.4]{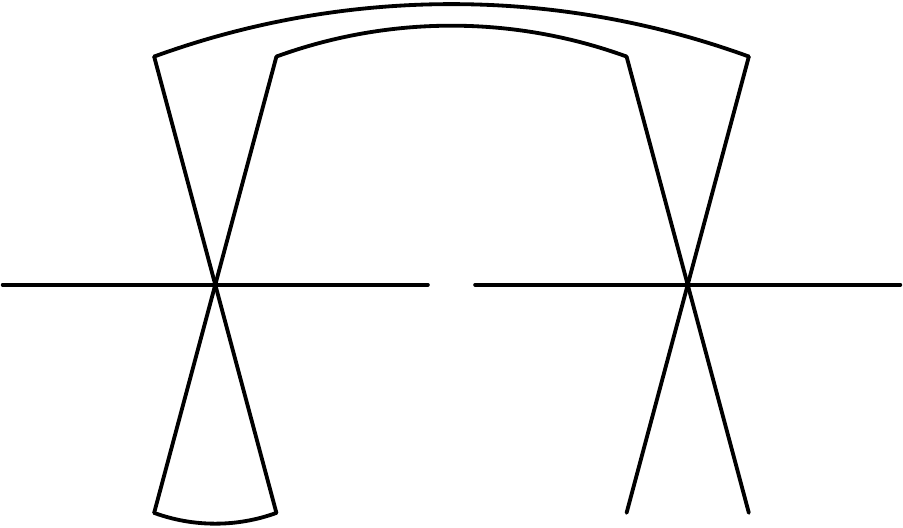}}\hspace{0.5cm}
  \caption{Face-connected divergent six-point graphs of $\vp_5^6$-model  }
  \label{fig:6-new}
\end{figure}
We introduce some formal sums:

%\begin{widetext}
\bea
&&
S^3 = \sum_{\stackrel{p_1,p_2,p_3}{q_1,q_2,q_3}}\big(\sum_{k=1}^3 p_k^2+(\sum_{k=1}^3 p_k)^2+m^2\big)^{-2}\big(\sum_{k=1}^3 q_k^2+(\sum_{k=1}^3 q_k)^2+m^2\big)^{-1} \crcr
&&
S^{13} =  \sum_{\stackrel{p_1,p_2,p_3}{q_1,q_2,q_3}}\big(\sum_{k=1}^3 p_k^2+(\sum_{k=1}^3 p_k)^2+m^2\big)^{-2}\big(\sum_{k=1}^3 q_k^2+(\sum_{k=1}^3 (p_k- q_k))^2 +(\sum_{k=1}^3p_k)^2 +m^2\big)^{-1} \crcr
&&
=   \sum_{\stackrel{p_1,p_2}{q_1,q_2,q_3,q_4}}\big(\sum_{k=1}^2 p_k^2+(\sum_{k=1}^4 q_k- \sum_{k=1}^2p_k)^2+(\sum_{k=1}^4 q_k)^2+m^2\big)^{-2}\big(\sum_{k=1}^4 q_k^2+(\sum_{k=1}^4 q_k)^2 +m^2\big)^{-1}.\cr\cr
&&
\eea
%\end{widetext}

A calculation yields, at low external momenta,
\bea
&&
A_{G_{1;\rho}}(0,\dots,0) = \frac{\lambda^2_{6,1;\rho}}{3^2.2!} K_{G_{1;\rho}}S^3  = 3\cdot 2. \lambda^2_{6,1;\rho} S^3,
\\
&&
A_{G_{2;\rho\rho'}}(0,\dots,0) = \frac{1}{3}\lambda_{6,1;\rho}
\sum_{\rho'}\lambda_{6,2;\rho\rho'} K_{G_{2;\rho\rho'}}  S^{13} = 3 \lambda_{6,1;\rho}
[\sum_{\rho'\neq \rho }\lambda_{6,2;\rho\rho'}]  S^{13},\\
&&A_{G'_{2;\rho\rho'}}(0,\dots,0) = \frac{1}{3}(\lambda_{6,1;\rho}
+\lambda_{6,1;\rho'}) \lambda_{6,2;\rho\rho'}K_{G'_{2;\rho\rho'}}  S^3 =2 (\lambda_{6,1;\rho} +\lambda_{6,1;\rho'} ) \lambda_{6,2;\rho\rho'} S^3,\\
&&
A_{G_{3;\rho\rho'}}(0,\dots,0)  = \lambda_{6,2;\rho\rho'}
\big[ \sum_{\tilde\rho \neq \rho}\lambda_{6,2;\rho\tilde\rho} +
\sum_{\tilde\rho \neq \rho'}\lambda_{6,2;\rho'\tilde\rho} \big] (S^3 +S^{13}) ,\\
&&K_{G_{1;\rho}} = 3^3 \cdot 2^2 ,\qquad K_{G_{2;\rho\rho'}} = 3 \cdot 3,\qquad K_{G'_{2;\rho\rho'}}=3 \cdot 2 . 
\eea
The contributions to $\Gamma_{6,1;\rho}$ are obtained from 
$G_{1;\rho}$ and  $G_{2;\rho\rho'}$.  Using these, we get 
\bea
&&
\Gamma_{6,1;\rho}(0,\dots ,0) = -\lambda_{6,1;\rho} + 
\lambda_{6,1;\rho} \Big[  6 \lambda_{6,1;\rho} S^3 +  3
[\sum_{\rho'\neq \rho }\lambda_{6,2;\rho\rho'}]  S^{13}\Big]
+  O(\lambda_{6,1}^p\lambda_{6,2}^{3-p}) \, .
\eea
The contributions to $\Gamma_{6,2;\rho\rho'}$ are obtained from 
$G'_{2;\rho\rho'}$ and $G_{3;\rho\rho'}$. One finds
\bea
\Gamma_{6,2;\rho\rho'}(0,\dots ,0)& =&
-  \lambda_{6,2;\rho\rho'}+ 2 (\lambda_{6,1;\rho} +\lambda_{6,1;\rho'} ) \lambda_{6,2;\rho\rho'} S^3 \cr
&+&\lambda_{6,2;\rho\rho'}
\big[ \sum_{\tilde\rho \neq \rho}\lambda_{6,2;\rho\tilde\rho} +
\sum_{\tilde\rho \neq \rho'}\lambda_{6,2;\rho'\tilde\rho} \big] (S^3 +S^{13})+ O(\lambda_{6,1}^p\lambda_{6,2}^{3-p})
\eea
 Reducing to the smaller space of couplings $\lambda_{6,1;\rho}=\lambda_{6,1}$ and $\lambda_{6,2;\rho\rho'}=\lambda_{6,2}$,
we get
%\begin{widetext}
\bea
\Gamma_{6,1}(0,\dots ,0) & =&-\lambda_{6,1} + 6 
\lambda_{6,1}^2  S^3 +  12\lambda_{6,1} \lambda_{6,2}  S^{13} +
 O(\lambda_{6,1}^p\lambda_{6,2}^{3-p}),
\crcr
\Gamma_{6,2}(0,\dots ,0) &=&
-  \lambda_{6,2}+  8\lambda_{6;2}^2(S^3 +S^{13})+  4\lambda_{6,2} \lambda_{6,1} S^{3}
+ O(\lambda_{6,1}^p\lambda_{6,2}^{3-p}).
\eea
Asymptotically, we can obtain the relation 
\bea
S^3
=\frac{\pi^3}{4}\,\mathcal{I'},\quad 
S^{13}=\frac{\pi^3}{\sqrt{31}}\,\mathcal{I'}
\eea
 (see Appendix II for more detail). At one-loop
the renormalized coupling constant $\lambda_{6,1}^{\rm ren}$ and  $\lambda_{6,2}^{\rm ren}$ are given by
\bea\label{ren1}
\lambda_{6,1}^{\rm ren}
&=&\lambda_{6,1}+\frac{9\pi^3}{4}\lambda_{6,1}^2\,\mathcal{I'}+12\big(\frac{49}{31\sqrt{31}}+\frac{5}{8}\big)\pi^3\lambda_{6,1}\lambda_{6,2}\mathcal{I'}
+O(\lambda_{6,1}^{p}\lambda_{6,2}^{3-p}),
\eea
and
\bea\label{ren2}
\lambda_{6,2}^{\rm ren}
&=&\lambda_{6,2}+4\Big(\frac{178}{31\sqrt{31}}+\frac{11}{8}\Big)\pi^3\lambda_{6,2}^2\,\mathcal{I'}+\frac{11\pi^3}{4}\lambda_{6,1}\lambda_{6,2}\,\mathcal{I'}
+O(\lambda_{6,1}^{p}\lambda_{6,2}^{3-p}).
\eea
$\square$
%\end{widetext}

 {\bf Discussion:}\,\,

$\bullet$  Let us come back on the subtle issue about
the notion of connectedness in this theory. The correct 
notion of connectedness should be the one of face-connectedness. 
Several graphs which a priori are divergent should not renormalize any coupling constant. For instance, graphs of the form given in \cref{fig:6-new2} are face-disconnected divergent six-point graphs. They do not contribute to the 1PI six-point functions. The amplitudes of the graphs are  

\bea
&&
A_{G''_{2;\rho\rho'}}(0,\dots,0) = \frac{1}{3}\lambda_{6,1;\rho}
\sum_{\rho'}\lambda_{6,2;\rho\rho'} K_{G''_{2;\rho\rho'}}  S^3 
= 3 \lambda_{6,1;\rho}
[\sum_{\rho'\neq \rho }\lambda_{6,2;\rho\rho'}] S^3 
\\
&& K_{G''_{2;\rho\rho'}} =3 \cdot 3 
\eea

\begin{figure}[htbp]
$G''_{2;\rho\rho'}${{\label{ClosedExttt}}\includegraphics[scale=0.4]{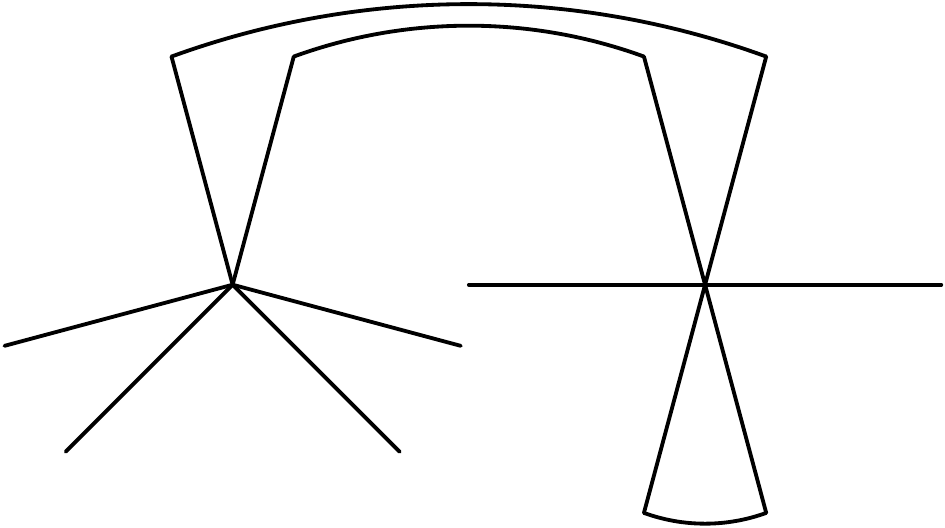}}\hspace{0.5cm}
  \caption{ Face-disconnected and divergent six-point graphs of $\vp_5^6$-model  }
  \label{fig:6-new2}
\end{figure} 
$\bullet$  We now discuss the results of Theorem \ref{prop3}.  Equation (\ref{ren1}) can be re-expressed  as  
\bea
\lambda_{6,1}^{\rm ren}=\lambda_{6,1}-\beta_{1}\lambda_{6,1}^2\,\mathcal{I'}-\beta_{12}\lambda_{6,1}\lambda_{6,2}\,\mathcal{I'}+O(\lambda_{6,1}^{p}\lambda_{6,2}^{3-p}),
\eea
where, at this order of perturbation, the $\beta$-function splits into coefficients $\beta_{1}$ and $\beta_{12}$ given by
\bea\label{b1}
\beta_1=-\frac{9\pi^3}{4},\quad \beta_{12}=-12\big(\frac{49}{31\sqrt{31}}+\frac{5}{8}\big)\pi^3.
\eea
This clearly shows that $\lambda_{6,1}^{\rm ren}\geq \lambda_{6,1}$ 
proving that this sector is asymptotically free, provided all 
coupling are positive. In the same way,
equation (\ref{ren2}) can be re-expressed  as   
\bea
 \lambda_{6,2}^{\rm ren}=\lambda_{6,2}-\beta_{2}\lambda_{6,2}^2\,\mathcal{I'}-\beta_{21}\lambda_{6,1}\lambda_{6,2}\,\mathcal{I'}+O(\lambda_{6,1}^{p}\lambda_{6,2}^{3-p})
\eea 
where the $\beta$-functions $\beta_2$ and $\beta_{21}$ are given by
\bea\label{b2}
\beta_{2}=-4\Big(\frac{178}{31\sqrt{31}}+\frac{11}{8}\Big)\pi^3,\quad \beta_{21}=-\frac{11\pi^3}{4}.
\eea
 
The same conclusion holds for the sector $\lambda_{6,2}$ which 
is asymptotically free. 
Both relations (\ref{b1}) and (\ref{b2}) show that the model with both interactions is asymptotically free in the $UV$ regime.  Hence, gauge invariant TGFT models of the form present here make a sense 
at arbitrary small scales yielding, far in the UV, a theory of non interacting spheres. Indeed, according to \cite{Gurau:2010ba}, all interactions presented here (called melonic) are nothing simplicial complexes with the sphere topology. 
The present results also show that both models might experience a phase transition when the renormalized coupling constants become larger and larger in the IR. This feature deserves full investigation. 

$\bullet$ We will now discuss a renormalized  coupling constants  $\lambda_{4,1}^{\rm ren}$ and $\lambda_{4,2}^{\rm ren}.$ We have already shown that, at high scale, the bare values of  coupling constants $\lambda_{6,1}$ and $\lambda_{6,2}$ vanish. Further the divergent four-point functions 
must not have more than one vertex type $V_{4,1},$ or $ V_{4,2}$, the only
divergent graphs are those couples with 
$V_{6,1}$ or $V_{6,2}$. Using relations (\ref{ren1}) and (\ref{ren2}) we come to the conclusion that
\bea
&&\lambda_{4,1}^{\rm ren}=\lambda_{4,1}+O(\lambda_{4,1}^{p}
\lambda_{6,k}^{3-p}),\quad k=1 \mbox{ or } k=2,\\
&&
\lambda_{4,2}^{\rm ren}=\lambda_{4,2}+O(\lambda_{4,2}^{p}
\lambda_{6,k}^{3-p}),\quad k=1 \mbox{ or } k=2.
\eea
Then the $\vp^4$ sector is safe at all loops and the $\beta$-functions  are given by
\bea
\beta_{4,1}=\beta_{4,2}=0.
\eea 

\section{Acknowledgements}
The author is   grateful to Vincent Rivasseau, Joseph Ben Geloun and Fabien Vignes-Tourneret for  useful comments that allowed to improve the paper.
This work is partially supported by the Abdus Salam International
Centre for Theoretical Physics (ICTP, Trieste, Italy) through the
Office of External Activities (OEA) - \mbox{Prj-15}. The ICMPA
is also in partnership with
the Daniel Iagolnitzer Foundation (DIF), France.

\section*{Appendix I:  Divergent series for $\vp_6^4$-model }\label{serie1}
\begin{proposition} Let $\mathcal{I}=\int_0^\infty \,d\alpha\,\frac{e^{-\alpha m^2}}{\alpha}$ be a logarithmically divergent quantity in the UV regime. The series $\mS'(0)$ and $\mathcal{K}(0)$  asymptotically write as 
\bea\label{zoulounew}
\mS'(0)=\frac{\pi^2}{\sqrt{5}}\,\mathcal{I},\quad \mathcal{K}(0)=\frac{\pi^2}{5\sqrt{5}}\,\mathcal{I}.
\eea
\end{proposition}
The rest of this section is devoted to the proof of this proposition.
Let us recall the Schwinger formula: Let $A$ be a positive define operator and $n$ is an integer then we  get
\bea\label{Wing}
\frac{1}{A^{n+1}}=\frac{1}{n!}\int_0^\infty\,d\alpha \, \alpha^ne^{-\alpha A}.
\eea
For 
$A=\sum_{k=1}^4 p_k^2+\big(\sum_{k=1}^4 p_k\big)^2+m^2$, we arrive at expression
\bea
\sum_{[p_{14}]\in\mathbb{Z}^4}\frac{1}{A^2}&=&\lim_{\Lambda\rightarrow 0}\lim_{\Lambda'\rightarrow 0}\sum_{[p_{14}]}^\Lambda\,\int_{\Lambda'}^\infty\,d\alpha \, \alpha^ne^{-\alpha A}\cr
&=&\lim_{\Lambda'\rightarrow 0}\int_{\Lambda'}^\infty\,d\alpha \, \alpha^n\lim_{\Lambda\rightarrow 0}\sum_{[p_{14}]}^\Lambda\,e^{-\alpha A}
\cr&=&\int_0^\infty\,d\alpha\,\alpha \,e^{-\alpha m^2}\sum_{[p_{14}]\in\mathbb{Z}^4}e^{-2\alpha[|p_{14}|^2+\sum_{i=1,\,i<j}^4p_ip_j]},
\eea
$|p_{14}|^2=\sum_{k=1}^4 p_k^2$, $[p_{ij}]=(p_i,p_{i+1},\cdots p_{j}).$  We have the following lemma
\begin{lemma}
Let $-\infty<p<\infty$. For $n\rightarrow \infty$, uniformaly in any finite interval of positive $\beta$,  we get
\bea
\sum_{p=-\infty}^\infty e^{-\frac{\beta}{n}p^2}=\sqrt{\frac{n\pi}{\beta}}.
\eea
\end{lemma}
\begin{proof}
The proof of this lemma is given  in  \cite{Hardy}.
\end{proof}
Noting that in the previous lemma $\frac{\beta}{n}\rightarrow 0$ as $\alpha= M^{-2i}\rightarrow 0$. Then $\sum_{p=-\infty}^\infty e^{-\alpha p^2}=\sqrt{\frac{\pi}{\alpha}}.$
Then
\bea
 \sum_{[p_{14}]\in\mathbb{Z}^4}e^{-2\alpha[|p_{14}|^2+\sum_{i=1,\,i<j}^4p_ip_j]}&=&\sqrt{\frac{\pi}{2\alpha}}\sqrt{\frac{2\pi}{3\alpha}}\sqrt{\frac{3\pi}{4\alpha}}\sqrt{\frac{4\pi}{5\alpha}}=\frac{\pi^2}{\alpha^2\sqrt{5}}.
\eea
 We arrive at the expression
\bea
\int_0^\infty\,d\alpha\,\alpha \,e^{-\alpha m^2}\sum_{[p_{14}]\in\mathbb{Z}^4}e^{-2\alpha[|p_{14}|^2+\sum_{i=1,\,i<j}^4p_ip_j]}=\frac{\pi^2}{\sqrt{5}}\int_{0}^{\infty}\, d\alpha \frac{e^{-\alpha m^2}}{\alpha}=\frac{\pi^2}{\sqrt{5}}\,\mathcal{I}.
\eea
Finally
$
\mS'(0)=\frac{\pi^2}{\sqrt{5}}\,\mathcal{I}.
$ Using the same argument, we get 
\bea
\frac{\big(\sum_{k=1}^4 p_k\big)^2}{A^3}&=&\frac{1}{2}\int_0^\infty\,d\alpha\,\alpha^2\,e^{-\alpha m^2} \sum_{[p_{14}]\in\mathbb{Z}^4}\Big(|p_{14}|^2+2\sum_{i=1,\,i<j}^4p_ip_j\Big)e^{-2\alpha[|p_{14}|^2+\sum_{i=1,\,i<j}^4p_ip_j]}\cr
&=&\frac{1}{2}\int_0^\infty\,d\alpha\,\alpha^2\,e^{-\alpha m^2} \sum_{[p_{14}]\in\mathbb{Z}^4}\Big(|p_{14}|^2+\sum_{i=1,\,i<j}^4p_ip_j\Big)e^{-2\alpha[|p_{14}|^2+\sum_{i=1,\,i<j}^4p_ip_j]}\cr
&+&\frac{1}{2}\int_0^\infty\,d\alpha\,\alpha^2 \sum_{[p_{14}]\in\mathbb{Z}^4}\sum_{i=1,\,i<j}^4p_ip_je^{-2\alpha[|p_{14}|^2+\sum_{i=1,\,i<j}^4p_ip_j]}\cr
&=&-\frac{1}{4}\int_0^\infty\,d\alpha\,\alpha^2 \,e^{-\alpha m^2}\frac{\partial}{\partial\alpha}\sum_{[p_{14}]\in\mathbb{Z}^4}e^{-2\alpha[|p_{14}|^2+\sum_{i=1,\,i<j}^4p_ip_j]}\cr
&+&\frac{1}{2}\int_0^\infty\,d\alpha\,\alpha^2 \,e^{-\alpha m^2}\sum_{[p_{14}]\in\mathbb{Z}^4}\sum_{i=1,\,i<j}^4p_ip_je^{-2\alpha[|p_{14}|^2+\sum_{i=1,\,i<j}^4p_ip_j]}\cr
&=&X_1+X_2,
\eea
with
\bea
X_1&=&-\frac{1}{4}\int_0^\infty\,d\alpha\,\alpha^2 \,e^{-\alpha m^2}\frac{\partial}{\partial\alpha}\sum_{[p_{14}]\in\mathbb{Z}^4}e^{-2\alpha[|p_{14}|^2+\sum_{i=1,\,i<j}^4p_ip_j]}\\
X_2&=&\frac{1}{2}\int_0^\infty\,d\alpha\,\alpha^2 \,e^{-\alpha m^2}\sum_{[p_{14}]\in\mathbb{Z}^4}\sum_{i=1,\,i<j}^4p_ip_je^{-2\alpha[|p_{14}|^2+\sum_{i=1,\,i<j}^4p_ip_j]},
\eea
and we get
\bea
X_1=-\frac{1}{4}\int_0^\infty\,d\alpha\,\alpha^2  \,e^{-\alpha m^2}\frac{\partial}{\partial\alpha}\Big(\frac{\pi^2}{\alpha^2\sqrt{5}}\Big)=\frac{\pi^2}{2\sqrt{5}}\,\mathcal{I}.
\eea
To compute $X_2$ let us give the following lemma
\begin{lemma}\label{xxx1}
Let $-\infty<p<\infty$. For $\alpha \rightarrow 0$ uniformly in any finite interval of constant $c$,  we get
\bea
\sum_{p=-\infty}^\infty p\,e^{-\alpha p^2+2c p}=\frac{c}{\alpha}\sqrt{\frac{\pi}{\alpha}}e^{\frac{\alpha^2}{c}},\quad \sum_{p=-\infty}^\infty p^n\,e^{-\alpha p^2+2c p}=\frac{1}{2^{n-1}\alpha}\sqrt{\frac{\pi}{\alpha}}\frac{d}{dc}\Big(c e^{\frac{\alpha^2}{c}}\Big).
\eea
\end{lemma}
Using this lemma we get easily 
\bea
X_2=\frac{1}{2}\int_0^\infty\,d\alpha\,\alpha^2 \,e^{-\alpha m^2}\Big(-\frac{3\pi^2}{5\alpha^3\sqrt{5}}\Big)=-\frac{3\pi^2}{10\sqrt{5}}\,\mathcal{I}.
\eea
Therefore 
$
\mathcal{K}(0)=\frac{\pi^2}{5\sqrt{5}}\,\mathcal{I}
$.

\section*{Appendix II: Divergent series for $\vp_5^6$-model }\label{serie2}
In this section, we will focus on the divergent terms of the $\vp_5^6$-model. Let us consider the  functions  $\Omega_{6,1}(b)$ and $\Omega_{6,2}(b,b')$.
 The second order partial derivative respect to external strand $b$    participated to the expression of the wave function. The  goal of this part is the proof of the following proposition
\begin{proposition}
Let $\,\mathcal{I'}=\int_0^\infty\,\int_0^\infty\,d^2\alpha\frac{e^{-2\alpha m^2}}{\alpha^2}$ be a logarithmically divergent quantity in the UV regime. The partial derivative of  $\Omega_{6,1}(b)$ and  $\Omega_{6,2}(b,b')$  are respectively given by
\bea
\frac{\partial^2}{\partial b^2}\Omega_{6,1}(b)\big|_{b=0}&=&\frac{5\pi^3}{4}\lambda_{6,1;1}\,\mathcal{I'},\\
\frac{\partial^2}{\partial b^2}\Omega_{6,2}(b,b')\big|_{b=b'=0}
&=&\big[\frac{80}{31\sqrt{31}}+\frac{5}{8}\big]\pi^3\lambda_{6,2;1\rho'}\,\mathcal{I'}.
\eea
\end{proposition}
The rest of this section is devoted to the proof of the above proposition. We have
\bea
\frac{\partial^2}{\partial b_1^2}\Omega_{6,1}(b)\big|_{b=0}&=&8\lambda_{6,1;1}\Big\{\sum_{\stackrel{p_1,p_2,p_3}{q_1,q_2,q_3}}\Big[\Big(\frac{1}{\chi_{(3)}^2(p)}-2\frac{(\sum p_k)^2}{\chi_{(3)}^3(p)}\Big)\Big]\Big[\frac{1}{\chi_{(3)}(q)}\Big]
-\frac{\sum p_k}{\chi_{(3)}^2(p)}\frac{\sum q_k}{\chi_{(3)}^2(q)}\Big\},
\eea
where
$
\chi_{(n)}(p)=\sum_{k=1}^n p_k^2+(\sum_{k=1}^n p_k)^2+m^2.
$
By using the Schwinger formula (\ref{Wing}), we find
\bea
\sum_{[p]\in\mathbb{Z}^3}\sum_{[q]\in\mathbb{Z}^3}\Big(\frac{1}{\chi_{(3)}^2(p)}\frac{1}{\chi_{(3)}(q)}\Big)&=&\int_0^\infty\int_0^\infty\,\alpha\, \,d\alpha\,d\beta\,\sum_{[p]\in\mathbb{Z}^3}e^{-\alpha\chi_{(3)}(p)}\sum_{[q]\in\mathbb{Z}^3}e^{-\beta\chi_{(3)}(q)}\cr
&=&\int_0^\infty\int_0^\infty\,\alpha\, e^{-\alpha m^2}\,e^{-\beta m^2}\,d\alpha\,d\beta\,\sqrt{\frac{\pi}{2\alpha}}\sqrt{\frac{2\pi}{3\alpha}}\sqrt{\frac{3\pi}{4\alpha}}\sqrt{\frac{\pi}{2\beta}}\sqrt{\frac{2\pi}{3\beta}}\sqrt{\frac{3\pi}{4\beta}}\cr
&=&\frac{\pi^3}{4}\int_0^\infty\int_0^\infty\,\alpha\, \frac{e^{-\alpha m^2}}{\alpha^{\frac{3}{2}}}\,\frac{e^{-\beta m^2}}{\beta^{\frac{3}{2}}}\,d\alpha\,d\beta=\frac{\pi^3}{4}\int_0^\infty\int_0^\infty\,d^2\alpha\,\frac{e^{-2\alpha m^2}}{\alpha^2}\cr
&=&\frac{\pi^3}{4}\,\mathcal{I'}.
\eea
In the same manner, we get
\bea
\sum_{\stackrel{[p]\in\mathbb{Z}^3}{[q]\in\mathbb{Z}^3}}\frac{(\sum p_k)^2}{\chi_{(3)}^3(p)}\frac{1}{\chi_{(3)}(q)}&=&\frac{1}{2}\int_0^\infty\int_0^\infty\,\alpha^2\, e^{-\alpha m^2}\,e^{-\beta m^2}\,\mathcal{P}(\alpha,\beta)\,d\alpha\,d\beta,
\eea
where 
$$\mathcal{P}(\alpha,\beta)=\sum_{[p]\in\mathbb{Z}^3}(\sum p_k)^2e^{-\alpha(\chi_{(3)}(p)-m^2)}\sum_{[q]\in\mathbb{Z}^3}e^{-\beta(\chi_{(3)}(q)-m^2)}$$
 can be computed  by using the following results:
\bea
&&\sum_{[p]\in\mathbb{Z}^3}(\sum p_k)^2e^{-\alpha(\chi_{(3)}(p)-m^2)}=\frac{3}{8\alpha}\sqrt{\frac{\pi}{2\alpha}}\sqrt{\frac{2\pi}{3\alpha}}\sqrt{\frac{3\pi}{4\alpha}},\\ &&\sum_{[q]\in\mathbb{Z}^3}e^{-\beta(\chi_{(3)}(q)-m^2)}=\sqrt{\frac{\pi}{2\beta}}\sqrt{\frac{2\pi}{3\beta}}\sqrt{\frac{3\pi}{4\beta}}.
\eea
 We  get
$
\sum_{\stackrel{[p]\in\mathbb{Z}^3}{[q]\in\mathbb{Z}^3}}\frac{(\sum p_k)^2}{\chi_{(3)}^3(p)}\frac{1}{\chi_{(3)}(q)}=\frac{3\pi^3}{64}\mathcal{I}.
$
A simple routine checking shows that \bea\label{alph}
\sum \frac{p_k}{\chi_{(3)}^2(p)}\sum\frac{ q_k}{\chi_{(3)}^2(q)}=0.
\eea
Finally
$$
\frac{\partial^2}{\partial b^2}\Omega_{6,1}(b)\big|_{b=0}=\frac{5\pi^3}{4}\lambda_{6,1;1}\,\mathcal{I}.
$$
 The second order partial derivative of $\Omega_{6,2}(b,b')$ respect to the external strand $b$ is written as
\bea
&&\frac{\partial^2}{\partial b^2}\Omega_{6,2}(b,b')\big|_{b=b'=0}=4\lambda_{6,2;1\rho'}\Big\{\sum_{\stackrel{p_1,p_2,p_3,p_4}{q_1,q_2}}\frac{1}{\chi_{(4)}(p)}\Big[\frac{1}{\chi_{(2,4)}^2(q,p)}
\cr
-&&2\frac{(\sum_{k=1}^2q_k-\sum_{k=1}^4p_k)^2}{\chi_{(2,4)}^3(q,p)}\Big]
+\sum_{\stackrel{p_1,p_2,p_3}{q_1,q_2,q_3}}\Big[\frac{1}{\chi^2_{(3)}(p)}-2\frac{(\sum p_k)^2}{\chi^3_{(3)}(p)}\Big]\frac{1}{\chi_{(3)}(q)}\Big\},
\eea
where 
$$\chi_{(m,n)}(q,p)=\sum_{k=1}^m q_k^2+(\sum_{k=1}^m q_k)^2+2(\sum_{k=1}^n p_k)^2-2\sum_{k=1}^n p_k\sum_{k=1}^m q_k+m^2.$$
 Let us compute the series $\sum_{p\in\mathbb{Z}^4}\sum_{q\in\mathbb{Z}^2}\frac{1}{\chi_{(2,4)}^2(q,p)}\frac{1}{\chi_{(4)}(p)}$\, and \, $\sum_{p\in\mathbb{Z}^4}\sum_{q\in\mathbb{Z}^2}\frac{1}{\chi_{(4)}(p)}
\frac{(\sum_{k=1}^2q_k-\sum_{k=1}^4p_k)^2}{\chi_{(2,4)}^3(q,p)}.$
Using the Schwinger formula we can write that
\bea
\sum_{p\in\mathbb{Z}^4}\sum_{q\in\mathbb{Z}^2}\frac{1}{\chi_{(2,4)}^2(q,p)}\frac{1}{\chi_{(4)}(p)}&=&\int_0^\infty\int_0^\infty\,\alpha\, e^{-\alpha m^2}\,e^{-\beta m^2}\,d\alpha\,d\beta\,\mathcal{Q}(\alpha,\beta),
\eea
where 
$$\mathcal{Q}(\alpha,\beta)=\sum_{[p]\in\mathbb{Z}^4}\sum_{[q]\in\mathbb{Z}^2}e^{-\alpha(\chi_{(2,4)}(q,p)-m^2)}e^{-\beta(\chi_{(4)}(p)-m^2)}.
$$
Now by lemma \ref{xxx1}, we reach
\bea
\sum_{[q]\in\mathbb{Z}^2}e^{-\alpha(\chi_{(2,4)}(q,p)-m^2)}=\sqrt{\frac{\pi}{2\alpha}}\sqrt{\frac{2\pi}{3\alpha}}
e^{-\frac{4}{3}\alpha(\sum_k p_k)^2}.
\eea
Moreover $\mathcal{Q}(\alpha,\beta)$ is given by
\bea
\mathcal{Q}(\alpha,\beta)=\sqrt{\frac{\pi}{2\alpha}}\sqrt{\frac{2\pi}{3\alpha}}\sum_{[p]\in\mathbb{Z}^4}
e^{-(\beta+\frac{4}{3}\alpha)(\sum_k p_k)^2-\beta|p_{14}|^2}=\sqrt{\frac{\pi}{2\alpha}}\sqrt{\frac{2\pi}{3\alpha}}
\sqrt{\frac{\pi}{a}}\sqrt{\frac{\pi}{a'}}\sqrt{\frac{\pi}{a''}}\sqrt{\frac{\pi}{a'''}},
\eea
where
\bea
a=2\beta+\frac{4}{3}\alpha, \,\,b=\beta+\frac{4}{3}\alpha,\,\, a'=a-\frac{b^2}{a},\,\,b'=-b+\frac{b^2}{a}\cr
a''=a'-\frac{b'^2}{a'},\,\,b''=b'+\frac{b'^2}{a'},\,\, a'''=a''-\frac{b''^2}{a''}.
\eea
Then for $\alpha=\beta$ we get 
\bea
a=\frac{10\alpha}{3},\,\, a'=\frac{17\alpha}{10},\,\, a''=\frac{24\alpha}{17},\,\, a'''=\frac{31\alpha}{24},
\eea
and
$$
\mathcal{Q}(\alpha,\alpha)=\sqrt{\frac{\pi}{2\alpha}}\sqrt{\frac{2\pi}{3\alpha}}
\sqrt{\frac{3\pi}{10\alpha}}\sqrt{\frac{10\pi}{17\alpha}}\sqrt{\frac{17\pi}{24\alpha}}\sqrt{\frac{24\pi}{31\alpha}}.
$$
Finally,
\bea
\sum_{p\in\mathbb{Z}^4}\sum_{q\in\mathbb{Z}^2}\frac{1}{\chi_{(2,4)}^2(q,p)}\frac{1}{\chi_{(4)}(p)}&=&\frac{\pi^3}{\sqrt{31}}\int_0^\infty\int_0^\infty\,d^2\alpha \frac{e^{-2\alpha m^2}}{\alpha^2}=\frac{\pi^3}{\sqrt{31}}\,\mathcal{I'}.
\eea
\bea
\sum_{p\in\mathbb{Z}^4}\sum_{q\in\mathbb{Z}^2}\frac{1}{\chi_{(4)}(p)}
\frac{(\sum_{k=1}^2q_k-\sum_{k=1}^4p_k)^2}{\chi_{(2,4)}^3(q,p)}&=&\frac{1}{2}\int_0^\infty\int_0^\infty\,d^2\alpha \,\alpha^2e^{-2\alpha m^2}\,\mathcal{R}(\alpha,\alpha)
\eea
where
$\mathcal{R}(\alpha,\alpha)=\sum_{p\in\mathbb{Z}^4}\sum_{q\in\mathbb{Z}^2}(\sum_{k=1}^2q_k-\sum_{k=1}^4p_k)^2e^{-\alpha(\chi_{(2,4)}(q,p)-m^2)}e^{-\alpha(\chi_{(4)}(p)-m^2)}.$ This quantity can be writen as
\bea
\mathcal{R}(\alpha,\alpha)&=&-\sum_{p\in\mathbb{Z}^4}\sum_{q\in\mathbb{Z}^2}\frac{\partial}{\partial \alpha}\Big(e^{-\alpha(\chi_{(2,4)}(q,p)-m^2)}\Big)e^{-\alpha(\chi_{(4)}(p)-m^2)}\cr
&-&\sum_{p\in\mathbb{Z}^4}\sum_{q\in\mathbb{Z}^2}\big(|q_{12}|^2+(\sum_k p_k)^2\big)e^{-\alpha(\chi_{(2,4)}(q,p)-m^2)}e^{-\alpha(\chi_{(4)}(p)-m^2)}\cr
&=&-\sum_{p\in\mathbb{Z}^4}\frac{\partial}{\partial \alpha}\Big(\sqrt{\frac{\pi}{2\alpha}}\sqrt{\frac{2\pi}{3\alpha}}
e^{-\frac{4}{3}\alpha(\sum_k p_k)^2}\Big)e^{-\alpha\big(|p_{14}|^2+(\sum_k p_k)^2\big)}\cr
&-&\sum_{p\in\mathbb{Z}^4}\sum_{q\in\mathbb{Z}^2}\big(|q_{12}|^2+(\sum_k p_k)^2\big)e^{-\alpha(\chi_{(2,4)}(q,p)-m^2)}e^{-\alpha(\chi_{(4)}(p)-m^2)}\cr
&=& \mathcal{R}_1+\mathcal{R}_2
\eea
where $$\mathcal{R}_1=-\sum_{p\in\mathbb{Z}^4}\frac{\partial}{\partial \alpha}\Big(\sqrt{\frac{\pi}{2\alpha}}\sqrt{\frac{2\pi}{3\alpha}}
e^{-\frac{4}{3}\alpha(\sum_k p_k)^2}\Big)e^{-\alpha\big(|p_{14}|^2+(\sum_k p_k)^2\big)}$$ and $$\mathcal{R}_2=-\sum_{p\in\mathbb{Z}^4}\sum_{q\in\mathbb{Z}^2}\big(|q_{12}|^2+(\sum_k p_k)^2\big)e^{-\alpha(\chi_{(2,4)}(q,p)-m^2)}e^{-\alpha(\chi_{(4)}(p)-m^2)}.$$ The additional contribution $\mathcal{R}_1$ can be evaluated as
\bea
\mathcal{R}_1&=&-\sum_{p\in\mathbb{Z}^4}\frac{\partial}{\partial \alpha}\Big(\sqrt{\frac{\pi}{2\alpha}}\sqrt{\frac{2\pi}{3\alpha}}
e^{-\frac{4}{3}\alpha(\sum_k p_k)^2}\Big)e^{-\alpha\big(|p_{14}|^2+(\sum_k p_k)^2\big)}\cr
&=&\sqrt{\frac{\pi}{2\alpha}}\sqrt{\frac{2\pi}{3\alpha}}\Big[\frac{1}{\alpha}\sum_{p\in\mathbb{Z}^4}e^{-\frac{7}{3}\alpha(\sum_k p_k)^2-\alpha|p_{14}|^2}+\frac{4}{3}\sum_{p\in\mathbb{Z}^4}(\sum_k p_k)^2e^{-\frac{7}{3}\alpha(\sum_k p_k)^2-\alpha|p_{14}|^2}\Big]\cr
&=&\mathcal{R}_{11}+\mathcal{R}_{12}.
\eea
In the above expression
$$\mathcal{R}_{11}=\frac{1}{\alpha}\sqrt{\frac{\pi}{2\alpha}}\sqrt{\frac{2\pi}{3\alpha}}\sqrt{\frac{3\pi}{10\alpha}}\sqrt{\frac{10\pi}{17\alpha}}\sqrt{\frac{17\pi}{24\alpha}}\sqrt{\frac{24\pi}{31\alpha}}=\frac{\pi^3}{\alpha^4\sqrt{31}}$$ and  $$\mathcal{R}_{12}=\frac{4}{3}\sqrt{\frac{\pi}{2\alpha}}\sqrt{\frac{2\pi}{3\alpha}}\sum_{p\in\mathbb{Z}^4}(\sum_k p_k)^2e^{-\frac{7}{3}\alpha(\sum_k p_k)^2-\alpha|p_{14}|^2}.$$  
We also have
\bea
U=\sum_{p\in\mathbb{Z}^4}(\sum_k p_k)^2e^{-\frac{7}{3}\alpha(\sum_k p_k)^2-\alpha|p_{14}|^2}=\frac{18\pi^2}{31\alpha^3\sqrt{93}}.
\eea
Therefore $\mathcal{R}_{12}=\frac{8\pi^3}{31\alpha^4\sqrt{31}}$ and then $\mathcal{R}_{1}=\frac{39\pi^3}{31\alpha^4\sqrt{31}}$. Using the same above  argument, we can prove that $\mathcal{R}_{2}=-\frac{28\pi^3}{31\alpha^4\sqrt{31}}$. Finally, it is straightforward to check  following  relation
\bea
\sum_{p\in\mathbb{Z}^4}\sum_{q\in\mathbb{Z}^2}\frac{1}{\chi_{(4)}(p)}
\frac{(\sum_{k=1}^2q_k-\sum_{k=1}^4p_k)^2}{\chi_{(2,4)}^3(q,p)}&=\frac{11\pi^3}{62\sqrt{31}}\,\mathcal{I'}.
\eea

\end{document}